\documentclass{article}
\usepackage{eurosym}
\usepackage{color}
\usepackage{graphicx}
\usepackage{graphics}
\usepackage{textcomp}
\usepackage{amsmath}
\usepackage{amssymb}
\usepackage{float}
\usepackage{epsfig}
\usepackage{epstopdf}
\usepackage{multirow}
\usepackage{array}
\usepackage{caption}
\usepackage{subcaption}
\usepackage{longtable}
\usepackage{pdflscape}
\usepackage{tabularx}
\usepackage{algorithm,algpseudocode}

\newtheorem{theorem}{Theorem}

\newtheorem{lemma}[theorem]{Lemma}

\newenvironment{proof}[1][Proof]{\noindent\textbf{#1.} }{\ \rule{0.5em}{0.5em}}

\begin{document}

\begin{center}
\bigskip \textbf{Efficient Construction of a Substitution Box Based on a Mordell
Elliptic Curve Over a Finite Field}\\[0pt]
\vspace{5mm} \textit{Naveed Ahmed Azam}$^{\ast }$,\textit{\ Umar Hayat%
\textnormal{,} Ikram Ullah} 

Department of Applied Mathematics and Physics, Graduate School of
Informatics, Kyoto University, Japan\\[0pt]
Email: $^{* }$azam@amp.i.kyoto-u.ac.jp

$^{* }$Corresponding author\\[0pt]
\end{center}

\textbf{Abstract}: Elliptic curve cryptography (ECC) is used in many security systems due to its small key size and high security as compared to the other cryptosystems. In many well-known security systems substitution box (S-box) is the only non-linear component. Recently, it is shown that the security of a cryptosystem can be improved by using dynamic S-boxes instead of a static S-box. This fact necessitates the construction of new secure S-boxes. In this paper, we propose an efficient method for the generation of S-boxes based on a class of Mordell elliptic curves (MECs) over prime fields by defining different total orders. The proposed scheme is developed in such a way that for each input it outputs an S-box in linear time and constant space. Due to this property, our method takes less time and space as compared to all existing S-box construction methods over elliptic curve. Furthermore, it is shown by the computational results that the proposed method is capable of generating cryptographically strong S-boxes with comparable security to some of the existing S-boxes constructed over different mathematical structures.


\section{Introduction}

Cryptography deals with the techniques to secure the private data. In these
techniques, the data is transformed into an unreadable form by using some
keys so that the adversaries cannot extract any useful information. S-box is
the only non-linear component of many well-known cryptosystems including Advanced Encryption System (AES).
It is therefore the security of such cryptosystems solely depends on the
cryptographic properties of their S-boxes. Shannon \cite{Shan} proved that a
cryptosystem is secure, if it can create confusion and diffusion in the data
up to a certain level. An S-box is cryptographically strong enough to create
desire confusion and diffusion, if it satisfies certain tests including the
test of non-linearity, approximation,
strict avalanche, bit independence and algebraic complexity.
Nowadays, AES is considered to be the most secured and widely used
cryptosystem, and hence many cryptographers studied its S-box. The study in
\cite{Thomas,Courtois, Murphy, Rosenthal} reveals that the AES S-box is vulnerable against
algebraic attacks because of its sparse polynomial representation. It is
also noticed that a cryptosystem based on a single S-box is unable to
generate desirable security, if the data is highly correlated \cite{Azam1,Hussain}. Furthermore, it is shown that the security of a cryptosystem can be
improved by using dynamic S-boxes instead of a static S-box, see for example \cite{Rahnama, Katiyar, Manjula, Agarwal, Balajee, Kazys}.
The two main reasons behind this are:
(a)~static S-box is vulnerable to data analysis attack and subkey attacks in which subkeys are obtained by using  inverse subbyte, if inverse of the S-box is known~\cite{Rahnama}; and
(b)~it is shown in ~\cite{Katiyar, Manjula, Agarwal, Balajee, Kazys}  that the algorithms using dynamic S-box are more complex and provide more overhead to the cryptanalysts when compared with static S-box.
Different image encryption algorithms by using dynamic S-box are presented in~\cite{Zaibi, Wangg, Devaraj,YeLiu}. In these studies it turned out that the image cryptosystems based on a dynamic S-box provide better security when compared with the cryptosystems using a static S-box. %
Due to these reasons many researchers have proposed new S-box generation
techniques based on different mathematical structures including algebraic,
and differential equations.

For an S-box design technique, it is necessary
that the resultant S-box: (a) inherits the properties of the underlying
mathematical structure. This is an important requirement which leads to the efficient generation and better understanding of the cryptographic properties of the S-box;
(b) is generated in low time and space complexity; and
(c) satisfies the security tests.
Of course, an S-box generation technique with high time complexity is not suitable for the cryptosystems using multiple, and dynamic S-boxes. Lui et al., \cite{Liuu} presented an improved AES S-box based
on an algebraic method. Cui et al., \cite{Cuii} used an affine function to
generate an S-box with 253 non-zero terms in its polynomial representation. Tran et al., \cite{Tran} used composition of a Gray code instead of an affine
mapping with the AES S-box to generate an S-box with high algebraic
complexity. Khan and Azam \cite{Khan1} proposed different methods for the generation of cryptographically strong S-boxes based on a
generalization of Gray S-box, and affine functions, see \cite{Khan2}. Azam \cite{Azam1} used
the later S-boxes for the encryption of confidential images. Chaotic maps
including Baker, logistic, and Chebyshev maps are used to generate new
S-boxes in \cite{Guoping,Guo,Neural}. Similarly, elliptic curves (ECs) are also used in
the field of cryptography for the development of highly secure
cryptosystems. Miller \cite{Miller} presented an EC based security system which
has smaller key size and higher security as compared to RSA. Jung et al., \cite{Jung} developed a link between the points on hyper-elliptic curves and
non-linearity of an S-box. Hayat et al., \cite{Hayat, umar2} for the first time used EC
over a prime field for the generation of dynamic S-boxes.
In these schemes, an S-box is generated by using the $x$-coordinates of the points on an ordered EC over a prime $p$, where the ordering $\prec$ on the points is performed with respect to their values i.e., for any two points $(x_1,y_1)$ and $(x_2,y_2)$ on the EC, $(x_1,y_1) \preceq (x_2,y_2)$, if $(y_{1}^{2} \leq y_{2}^{2})$ (mod $p$).
Actually, the scheme in \cite{umar2} is a generalization of the scheme in \cite{Hayat}.
Although these methods are capable of generating cryptographically strong S-boxes, but they have the next two weak points.
Firstly, they need to compute and store the EC for their generation process.
Due to this, the time and space complexity of these schemes are $\mathcal{O}(p^{2})$ and $\mathcal{O}(p)$, respectively, where $p\geq257$ is the prime of the underlying EC.
Secondly, the output of these schemes is uncertain i.e., for each set of input parameters the algorithms do not necessarily output
an S-box.

The purpose of this article is to develop such a novel and efficient S-box
generation technique based on a finite Mordell elliptic curve (MEC) which generates secure S-box inheriting the
properties of the underlying MEC for each set of input parameters.
To achieve this, we defined some typical
type of total orders on the points of the MEC, and then used $y$%
-coordinates instead of $x$-coordinates to obtain an S-box. The remaining
paper is organized as follows:
Some basic definitions and results
related to EC are discussed in Section 2.
The proposed algorithm is described in
Section 3. Section 4 contains the security analysis, while a detailed comparison of the newly developed scheme with some of the existing methods is performed in Section 5.
Finally, conclusions are drawn in Section 6.
\section{Preliminaries}


For a prime $p$, and two non-negative integers $a,b\leq p-1$, the EC $E_{p,a,b}$ over the prime field $\mathbf{F}%
_{p}$ is defined to be the collection of infinity point $O$, and all ordered
pairs $(x,y)\in \mathbf{F}_{p}\times \mathbf{F}_{p}$ satisfying the equation
\begin{equation*}
y^{2}\equiv x^{3}+ax+b\ (\text{mod}\hspace{1mm} p).
\end{equation*}%
We call $p,a$, and $b$ the parameters of the EC $E_{p,a,b}$. An approximation for
the  number of points $\#E_{p,a,b}$ on $E_{p,a,b}$ can be obtained by
using Hasse's formula \cite{WAS,Brown}
\begin{equation*}
\left\vert \#E_{p,a,b}-p-1\right\vert \leq 2\sqrt{p}.
\end{equation*}

Mordell elliptic curve (MEC) is a special kind of elliptic curve with $a=0$%
. The significance of some MECs $E_{p,0,b}$ is that they have exactly $%
p+1$ points.
The following Theorem \cite{WAS}  gives the information of such MECs.
\begin{theorem}
\label{mrdell} Let $p>3$ be a prime such that $p\equiv 2$ (mod 3). Then
for each $b\in \mathbf{F}_{p}$, the MEC $E_{p,0,b}$ has exactly $p+1$
distinct points, and has each integer in $[0,p-1]$ exactly once as a $y$%
-coordinate.
\end{theorem}
Henceforth, a MEC $E_{p,0,b}$, where $p\equiv 2$ (mod~3), is simply denoted by $E_{p\equiv 2,b}$.
\section{Description of the Proposed S-box Designing Technique}

In this section, we give an informal intuition of our proposed method. Our
aim is to develop such an S-box generation technique
based on a MEC which outputs an S-box: (a) in linear time and constant space for each set of input parameters; (b) that
inherits the properties of the underlying MEC; and (c) having high security against
cryptanalysis. Note that the S-box design techniques proposed by
\cite{Hayat, umar2} do not satisfy conditions (a) and (b).
One of the possible ways of
designing such a technique is to input such an EC which contains all integer
values from $[0, 255]$ without repetition.
It is, therefore, the proposed algorithm takes a MEC $E_{p\equiv 2,b}$ as an input, and uses $y$%
-coordinates to generate an S-box instead of $x$-coordinates.
The next task is to use these $y$%
-coordinates in such a way that the resultant S-box inherits the properties
of the underlying MEC. Of course, the usage of some arithmetic operations such
as modulo operation for this purpose S-box will destroy the
structure of the underlying MEC. Thus, we used the concept of total
order on the MEC to get an S-box. Order theory is intensively used in formal
methods, programming languages, logic, and statistic analysis. Now the
natural question is how to define different orderings on the MEC. Note
that for each $x$ value of MEC, there are two $y$ values. Thus, we can
divide the orderings on MEC into two categories: (1) one is that in which
the two $y$ values of each $x$ appear consecutively; and (2) the other one
contains those orderings in which the two $y $ values of each $x$ do not
appear consecutively. Based on this fact, we defined three different type of
orderings on a given MEC $E_{p\equiv 2,b}$ to generate three different S-boxes.
\subsection{The proposed orderings on a MEC $E_{p\equiv 2,b}$}

The orderings used in the proposed method are discussed below.

\textbf{(1) A natural ordering on a MEC: }We define a natural ordering $\prec
_{N}$ on $E_{p\equiv 2,b}$ based on $x$-coordinates as follows
\begin{equation}
(x_{1},y_{1})\prec _{N}(x_{2},y_{2})\Leftrightarrow \left\{
\begin{tabular}{l}
either if $x_{1}<x_{2};$ or \\
if $x_{1}=x_{2},$ and $y_{1}<y_{2}$,%
\end{tabular}%
\right.
\end{equation}%
where $(x_{1},y_{1}), (x_{2},y_{2}) \in E_{p\equiv 2,b}$.

The aim of this ordering is to sort the points
on the MEC in such a way that the $x$-coordinates are in non-decreasing order,
and the two $y$ values corresponding to each $x$ appear consecutively.

The next two orderings are introduced based on the following observation deduced from Theorem \ref{mrdell} to diffuse the $y$-coordinates on a MEC.

\bigskip
\textbf{Observation:} \label{x1notx2}For any two distinct points $(x_{1},y_{1})$ and $(x_{2},y_{2})
$ on the MEC $E_{p\equiv 2,b}$, and either $%
x_{1}+y_{1}=x_{2}+y_{2}$ or $x_{1}+y_{1}\equiv x_{2}+y_{2}(\text{mod}\hspace{1mm}p)$, it
holds that $x_{1}\neq x_{2}$.

\bigskip

\textbf{(2) A diffusion ordering on a MEC: }An ordering is defined on $E_{p\equiv 2,b}$ to diffuse the two $y$ values of each $x$. Let $(x_{1},y_{1})$
and $(x_{2},y_{2})$ be any two points on $E_{p\equiv 2,b}$, the diffusion
ordering $\prec _{D}$ is defined to be
\begin{equation}
(x_{1},y_{1})\prec _{D}(x_{2},y_{2})\Leftrightarrow \left\{
\begin{tabular}{l}
either if $x_{1}+y_{1}<x_{2}+y_{2};$ or \\
if $x_{1}+y_{1}=x_{2}+y_{2}$, and $x_{1}<x_{2}.$%
\end{tabular}%
\right.
\end{equation}

\begin{lemma}
\label{DO} The relation $\prec _{D}$ is a total order on the MEC $E_{p\equiv 2,b}$.
\end{lemma}

\begin{proof}
For each $(x_{1},y_{1})\in E_{p\equiv 2,b}$, we have $x_{1}+y_{1}=x_{1}+y_{1}$,
and therefore $(x_{1},y_{1})\prec _{D}(x_{1},y_{1})$. This implies that $%
\prec _{D}$ is reflexive. Next, we need to show that $\prec _{D}$ satisfies
the antisymmetric property. Thus, for $(x_{1},y_{1}),(x_{2},y_{2})\in
E_{p\equiv 2,b}$, suppose that $(x_{1},y_{1})\prec _{D}(x_{2},y_{2}),$ and $%
(x_{2},y_{2})\prec _{D}(x_{1},y_{1})$ hold. This implies that $%
x_{1}+y_{1}=x_{2}+y_{2}$. This is because of the fact that $%
x_{1}+y_{1}<x_{2}+y_{2}$, and $x_{2}+y_{2}<x_{1}+y_{1}$ are the only cases
for which the supposition and $x_{1}+y_{1}\neq x_{2}+y_{2}$ are true, which
eventually imply that $x_{1}+y_{1}=x_{2}+y_{2}.$ Now if $x_{1}\neq x_{2},$ then by
the supposition and the fact $x_{1}+y_{1}=x_{2}+y_{2}$, we have $x_{1}<x_{2}$ and $x_{2}<x_{1}$, which lead to the contradiction $x_{1}=x_{2}$.
Thus $x_{1}+y_{1}=x_{2}+y_{2}$ and $x_{1}=x_{2}$ hold, which ultimately
imply that $y_{1}=y_{2}$, and therefore $(x_{1},y_{1})=(x_{2},y_{2})$. Now,
to prove the transitivity property, suppose that $(x_{1},y_{1})\prec
_{D}(x_{2},y_{2})$, and $(x_{2},y_{2})\prec _{D}(x_{3},y_{3})$ hold, where $%
(x_{1},y_{1}),(x_{2},y_{2}),(x_{3},y_{3})\in E_{p\equiv 2,b}.$ Now if $%
x_{1}+y_{1}<x_{2}+y_{2}$ and $x_{2}+y_{2}\leq x_{3}+y_{3}$, or $%
x_{1}+y_{1}=x_{2}+y_{2}$ and $x_{2}+y_{2}<x_{3}+y_{3}$, then $%
x_{1}+y_{1}<x_{3}+y_{3}$, and therefore $(x_{1},y_{1})\prec
_{D}(x_{3},y_{3}).$ Similarly, if $x_{1}+y_{1}=x_{2}+y_{2}=x_{3}+y_{3}$, then $x_{1}<x_{2}$ and $x_{2}<x_{3}$, and hence $x_{1}+y_{1}=x_{3}+y_{3}$ and $%
x_{1}<x_{3}.$ This completes the proof.
\end{proof}
\bigskip
\\
\textbf{(3) A modulo diffusion ordering on a MEC: }The order $\prec _{M}$
defined below produces diffusion in both $x$-coordinates and $y$-coordinates
of the points on $E_{p\equiv 2,b}$. Let $(x_{1},y_{1}),(x_{2},y_{2})\in
E_{p\equiv 2,b}$, then
\begin{equation}
(x_{1},y_{1})\prec _{M}(x_{2},y_{2})\Leftrightarrow \left\{
\begin{tabular}{l}
either if $(x_{1}+y_{1}<x_{2}+y_{2})$(mod $p);$ or \\
if $x_{1}+y_{1}\equiv x_{2}+y_{2}$(mod $p)$, and $x_{1}<x_{2}.$%
\end{tabular}%
\right.
\end{equation}

\begin{lemma}
\label{mod}The relation $\prec _{M}$ is a total order on the MEC $E_{p\equiv 2,b}$.
\end{lemma}

Lemma \ref{mod} can be proved by using the similar arguments as used in the
proof of Lemma \ref{DO}.

The effect of these orderings $\prec _{N},\prec _{D}$ and $\prec _{M}$ on $y$%
-coordinates of the MEC $E_{101\equiv 2,1}$ is shown in Figure \ref{fig:D} by plotting them in a non-decreasing order of their points on the MEC w.r.t $\prec _{N},\prec _{D}$ and $\prec _{M}$, respectively.
\begin{figure}[H]
\centering
\includegraphics [scale=0.90]{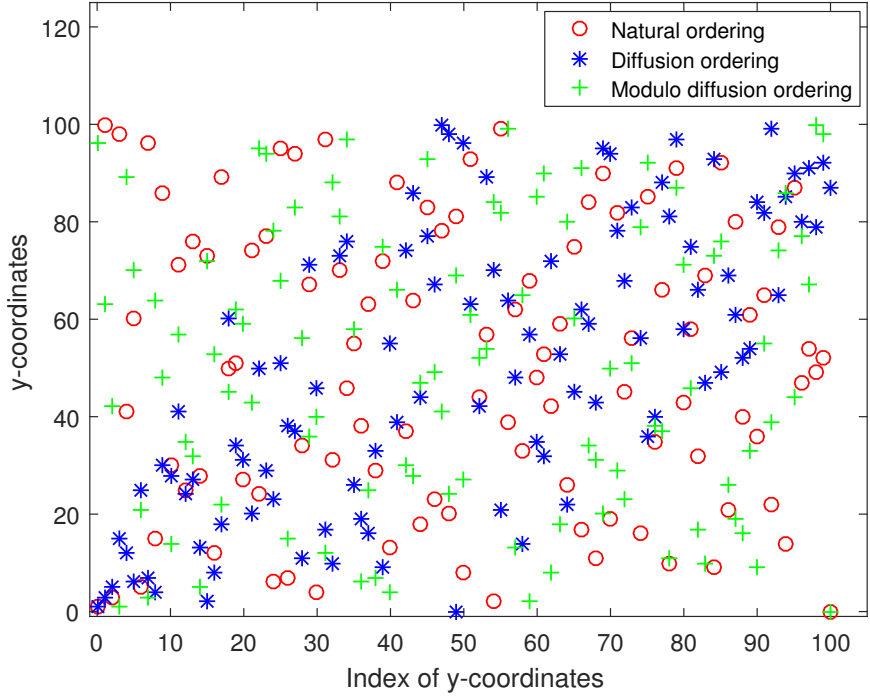}
\caption{The arrangements of $y$-coordinates of $E_{101\equiv 2,1}$ w.r.t. the
proposed orderings }
\label{fig:D}
\end{figure}
Similarly, a relation among the sets of all $y$-coordinates of the MEC $E_{p\equiv 2,b}$ obtained by different proposed orderings $\prec _{H}$ and $\prec _{K}$, where $H,K\in \{N,D,M\}$, is quantified by computing their correlation coefficient $\rho_{HK}$. The correlation results for different MECs are shown in Table \ref{Corr}. It is evident from the results that each ordering has different effect on the $y$-coordinates of the underlying MEC.

%
\begin{table}[H]
\caption{ Results of the correlation test}
\label{Corr}
\centering
\begin{tabular}{ccccc}\\[-3mm]
\hline
$p$    & $b$   & $\rho_{ND}$ & $\rho_{ND}$ & $\rho_{DM}$ \\
\hline
101  & 1   & -0.0588             & 0.0550              & -0.0497             \\
827  & 87  & -0.0044             & 0.0008              & 0.0027              \\
1013 & 118 & 0.0028               & -0.0059             & 0.0003              \\
2027 & 8   & 0.0007              & -0.0068             & -0.0002    \\
\hline
\end{tabular}
\end{table}
\subsection{The proposed S-box construction method}
Let $E_{p\equiv 2,b}$ be a
Mordell elliptic curve (MEC), where $p\geq 257$. The lower bound on the prime $p$ is $257$ for
the proposed method so that MEC has at least $256$ points. An S-box ${\small %
S}_{p,b}^{H}$, where $H\in \{N,D,M\}$, is generated by selecting the $y$%
-coordinates on $E_{p\equiv 2,b}$ which are in the interval $[0, 255]$ as
${\small S}_{p,b}^{H}\colon \{0,1,\ldots ,255\}\rightarrow \{0,1,\ldots ,255\} \text{ defined as } {\small S}_{p,b}^{H}(i)=y_{i}$,
such that $(x_{i},y_{i})\in E_{p\equiv 2,b}$, and \\ $(x_{i-1},y_{i-1}) \prec
_{H}(x_{i},y_{i})$.

It is clear from Theorem \ref{mrdell} that ${\small %
S}_{p,b}^{H}$ is a bijection, which further implies that the proposed method generates an S-box for each set of input parameters.
%
%
\begin{lemma}\label{complexity}
For any prime $p\geq 257$ such that $p\equiv 2$~(mod 3), integer $b\in
\lbrack 0,p-1]$, and $H\in \{N,D,M\},$ the S-box $S_{p,b}^{H}$ can be
generated in time complexity $\mathcal{O}(p)$ and constant space.
\end{lemma}
\begin{proof}
The generation of $S_{p,b}^{H}$ requires calculation of $256$ points on the MEC with $y$-coordinates in $[0,255]$, and then their sorting.
The calculation of $256$ points on the MEC can be done in $\mathcal{O}(p)$, since for each $y\in [0, 255]$, a for loop of size $p$ is required to find integer $x$ such that $(x, y)$ is a point on the MEC.
However, the sorting of these $256$ points can be done in a constant time with respect to the ordering $H$.
Thus, $S_{p,b}^{H}$ can be generated in $\mathcal{O}(p)$ time.
Furthermore, the generation process store only 256 points on the MEC for sorting purpose, and therefore it takes constant space.
\end{proof}

It is evident from Lemma~\ref{complexity} that the time and space complexity of the proposed S-box generation method is independent of the parameter $b$ and the ordering on the underlying MEC.
An algorithmic description of the proposed generation method is given in Algorithm~1.

\begin{algorithm}[H]
\caption{{The proposed S-box generation method}}
\begin{algorithmic}[1]
\Require  A Mordell elliptic curve $E_{p ,b}$, where $p \equiv 2$ (mod 3), with a total order $H\in \{N,D,M\}$.
\Ensure The proposed S-box $S_{p,b}^{H}$.
\State{$A := \emptyset$; /* The set of 256 points of the MEC with $y$-coordinates in [0, 255]*/}
\For{each $y = 0, 1, \ldots, 255 $}
	\While{$x \in [0, p-1] $}
		\If{$x^{3} + b \equiv y^{2}$ (mod $p$)}
			\State{$A := A \cup \{(x, y)\}$}
		\EndIf
	\EndWhile
\EndFor
\State{Sort $A$ with respect to the ordering $H$;}
\State{{Output} all $y$-coordinates of the points in $A$ preserving their order as the S-box $S_{p,b}^{H}$.}

\end{algorithmic}
\end{algorithm}

The S-boxes ${\small S}%
_{1667,351}^{N},{\small S}_{3299,1451}^{D}$ and ${\small S}_{4229,2422}^{M}$
generated by the proposed technique are presented in Tables (\ref{SN})-(\ref%
{SM}), respectively.

\section{Security Analysis}

Several standard tests are applied on the S-boxes obtained by the proposed
method to test their cryptographic strength. A brief introduction to these security tests, and their results for some of the newly generated
S-boxes ${\small S}_{1667,351}^{N}$, ${\small S}_{1949,544}^{N}$, ${\small S}%
_{3023,626}^{N},$ ${\small S}_{3299,1451}^{D}$, ${\small S}_{3041,1298}^{D}$%
, ${\small S}_{3347,2937}^{D}$, ${\small S}_{4229,2422}^{M}$, ${S}%
_{4217,1156}^{M}$ and ${\small S}_{3299,1400}^{M}$ are discussed in this
section.


\subsection{Non-Linearity (NL)}

It is important for an S-box to create confusion in the data up to a certain
level to keep the data secure from the adversaries. The {\normalsize %
confusion creation capability of an S-box }$S$ over the Galois Field $%
GF(2^{8})$ {\normalsize is measured by its non-linearity }$\mathcal{N}(S)$,
which is defined below
\begin{equation*}
\mathcal{N}(S)=\min_{\alpha ,\beta ,\gamma }\{x\in GF(2^{8}):\alpha \cdot
S(x)\neq \beta \cdot x\oplus \gamma \},
\end{equation*}%
where $\alpha \in GF(2^{8})${\normalsize , }$\gamma \in GF(2)$%
{\normalsize , }$\beta \in GF(2^{8})\backslash \{0\}${\normalsize \ and
\textquotedblleft }$\cdot ${\normalsize \textquotedblright\ represents dot
product over }$GF(2).$

An S-box with high NL is capable of
generating high confusion in the data. However, it is also shown in \cite{Willi} that an S-box with high NL may not satisfy other cryptographic
properties. The NL of some of the newly constructed S-boxes is listed in
Table \ref{NL}. Note that each listed S-box has NL 106, which is large
enough to create high confusion.
\begin{table}[H]
\caption{Non-linearity of the newly generated S-boxes}
\label{NL}{\small \centering
\scalebox{.70} {\begin{tabular}{llllllllllllllll}
\hline \\[-3mm]
S-boxes	&	$S_{1667,351}^{N}$	&	$S_{1949,544}^{N}$	&	$S_{3023,626}^{N}$	&	$S_{3299,1451}^{D}$	&	$S_{3041,1298}^{D}$	&	$S_{3347,2937}^{D}$	&	$S_{4229,2422}^{M}$	&	$S_{4217,1156}^{M}$	&	$S_{3299,1400}^{M}$	\\
\hline																			
NL	&	106	&	106	&	106	&	106	&	106	&	106	&	106	&	106	&	106	\\
 \hline
\end{tabular}}  }
\end{table}

\subsection{\protect\normalsize Approximation Attacks}

{\normalsize A cryptographically strong S-box must have high resistance
against approximation attacks. The approximation attacks can be divided into
two categories namely linear approximation attacks, and differential
approximation attacks which are explained below.}

The resistance of an S-box $S$ against linear approximation attacks is
measured by calculating its maximum number $\mathcal{L(}S\mathcal{)}$ of
coincident input bits with the output bits. The mathematical expression of $%
\mathcal{L(}S\mathcal{)}$ is as follows%
\begin{equation*}
{\small \mathcal{L(}S\mathcal{)}{\normalsize =}\frac{1}{2^{8}}\left\{
{\normalsize \max_{\alpha ,\beta }}\left\{ \left\vert {\normalsize \#}%
\left\{ {\normalsize x\in GF(2^{8}):\alpha \cdot x=\beta \cdot S(x)}\right\}
{\normalsize -2^{7}}\right\vert \right\} \right\} ,}
\end{equation*}%
where $\alpha \in GF(2^{8})$ and\ $\beta \in GF(2^{8})\backslash \{0\}$.

An S-box $S$ is said to be highly resistive against linear approximation
attacks if it has low value of $\mathcal{L(}S\mathcal{)}$. The LAP of the
newly generated S-boxes is listed in Table \ref{LAP}. The average LAP of all
of the listed S-boxes is $0.1371$ which is very low, and hence the proposed
scheme is capable of generating S-boxes with high resistance against
linear approximation attacks.
\begin{table}[]
\caption{LAP of the newly generated S-boxes}
\label{LAP}{\small \centering
\scalebox{0.70} {\begin{tabular}{llllllllllllllll}
\hline \\[-3mm]
S-boxes	&	$S_{1667,351}^{N}$	&	$S_{1949,544}^{N}$	&	$S_{3023,626}^{N}$	&	$S_{3299,1451}^{D}$	&	$S_{3041,1298}^{D}$	&	$S_{3347,2937}^{D}$	&	$S_{4229,2422}^{M}$	&	$S_{4217,1156}^{M}$	&	$S_{3299,1400}^{M}$	\\
\hline																			
LAP	&	0.1328	&	0.1328	&	0.1406	&	0.1484	&	0.1328	&	0.1406	&	0.1328	&	0.1328	&	0.1406	\\
\hline
\end{tabular}} }
\end{table}

\subsubsection{Differential Approximation Probability (DAP)}

The strength of an S-box against differential approximation attacks is
measured by calculating its DAP. For an S-box $S$, the DAP $\mathcal{D}(S)$
is the maximum probability of a specific change ${\normalsize \triangle y}$
in the output bits ${\normalsize S(x)}$ when the input bits ${\normalsize x}$
are changed to ${\normalsize x\oplus \triangle x}$ i.e.,
\begin{equation*}
\mathcal{D}(S){\normalsize =}\frac{1}{2^{8}}\left\{ {\normalsize %
\max_{\triangle x,\triangle y}}\left\{ {\normalsize \#}\left\{ {\normalsize %
x\in GF(2}^{8}{\normalsize ):S(x\oplus \triangle x)=S(x)\oplus \triangle y}%
\right\} \right\} \right\} {\normalsize ,}
\end{equation*}%
where $\triangle x,$\ $\triangle y\in GF(2^{8}),$ and \textquotedblleft $%
{\normalsize \oplus }$\textquotedblright\ is bit-wise addition over $GF(2)$.

The smaller is the value of DAP, the higher is the security of the S-box against
differential approximation attacks. The experimental results of DAP on the
newly generated S-boxes are presented in Table \ref{dap}. It is evident from
Table \ref{dap} that the newly generated S-boxes have high resistance
against differential attacks.
\begin{table}[H]
\caption{DAP of the newly generated S-boxes}
\label{dap}{\small \centering
\scalebox{0.70} {\begin{tabular}{llllllllllllllll}
\hline \\[-3mm]
S-boxes	&	$S_{1667,351}^{N}$	&	$S_{1949,544}^{N}$	&	$S_{3023,626}^{N}$	&	$S_{3299,1451}^{D}$	&	$S_{3041,1298}^{D}$	&	$S_{3347,2937}^{D}$	&	$S_{4229,2422}^{M}$	&	$S_{4217,1156}^{M}$	&	$S_{3299,1400}^{M}$	\\
\hline																			
DAP	&	0.0391	&	0.0391	&	0.0391	&	0.0391	&	0.0391	&	0.0391	&	0.0391	&	0.0391	&	0.0391	\\
\hline
\end{tabular}}  }
\end{table}
\subsection{Strict Avalanche Criterion (SAC)}
The diffusion creation capability of an S-box is calculated by SAC. The SAC
of an S-box $S$ is the measure of change in output bits when a single input
bit is changed. The SAC of an S-box $S$ with boolean functions $S_{i},$ where $1\leq i\leq 8$, is computed by calculating an eight dimensional square matrix $M(S)=[m_{ij}]$ by using each of the eight elements $%
{\normalsize \alpha }_{j}{\normalsize \in GF(2}^{8}{\normalsize )}$ with
only one non-zero bit as%
\begin{equation*}
{\small m_{ij}=\frac{1}{2^{8}}\left( \sum_{x\in GF(2^{8})}{\normalsize w}%
\left( {\normalsize S}_{i}{\normalsize (x\oplus \alpha }_{j}{\normalsize %
)\oplus S}_{i}{\normalsize (x)}\right) \right) {\normalsize ,}}
\end{equation*}%
where $w(v)$\ denotes the number of non-zero bits in the vector $v$.

SAC test is fulfilled, if all entries of $M(S)$ are close to $0.5$. The entries
of SAC matrix corresponding to each newly generated S-boxes ${\small S}%
_{1667,351}^{N}$, ${\small S}_{3299,1451}^{D}$ and ${\small S}%
_{4229,2422}^{M} $ are plotted in a linear order in Figure 2. The average of
minimum, and maximum values of $M(S)$ corresponding to each of the newly
generated S-boxes are $0.4115$ and $0.6094$, respectively. Table \ref{sac}
clearly shows that the S-boxes generated by the proposed method based on a
MEC is capable of generating high diffusion in the data. 

\begin{table}[H]
\caption{SAC of the newly generated S-boxes}
\label{sac}{\small \centering
\scalebox{0.70} {\begin{tabular}{llllllllllllllll}
\hline \\[-3mm]
S-boxes	&	$S_{1667,351}^{N}$	&	$S_{1949,544}^{N}$	&	$S_{3023,626}^{N}$	&	$S_{3299,1451}^{D}$	&	$S_{3041,1298}^{D}$	&	$S_{3347,2937}^{D}$	&	$S_{4229,2422}^{M}$	&	$S_{4217,1156}^{M}$	&	$S_{3299,1400}^{M}$	\\
\hline																			
SAC(max)	&	0.5938	&	0.625	&	0.6563	&	0.6406	&	0.6094	&	0.6094	&	0.5938	&	0.6094	&	0.625	\\
SAC(min)	&	0.4531	&	0.4219	&	0.4219	&	0.4063	&	0.4219	&	0.4063	&	0.375	&	0.3906	&	0.3594	\\
\hline
\end{tabular}}  }
\end{table}

\begin{figure}[H]
\label{fig:sac}\centering
\includegraphics [scale=0.80]{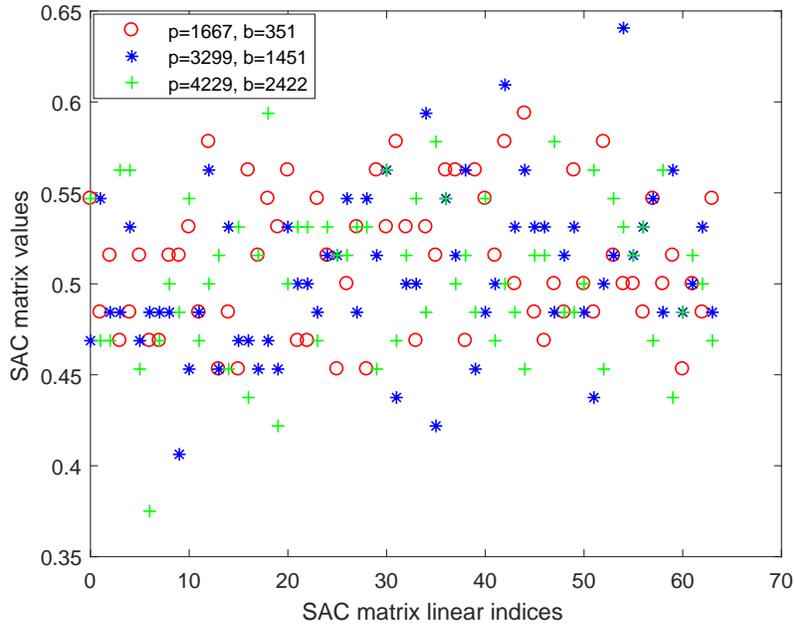}
\caption{SAC matrix plot for ${\protect\small S}_{1667,351}^{N}$, $%
{\protect\small S}_{3299,1451}^{D}$ and ${\protect\small S}_{4229,2422}^{M}$}
\end{figure}

\subsection{Bit Independence Criterion (BIC)}

BIC is also an important test to measure the diffusion creation strength of
an S-box. The main idea of this test is to investigate the dependence of a
pair of output bits when an input bit is reversed. The BIC of an S-box $S$
over $GF(2^{8})$\ with $S_{i}$ boolean functions is also calculated by
computing a square matrix $N(S)=[n_{ij}]$ of dimension eight as follows
\begin{equation*}
{\small n_{ij}=\frac{1}{2^{8}}\left( \sum_{\substack{ x\in GF(2^{8})  \\ %
1\leq k\leq 8}}{\normalsize w}\left( {\normalsize S}_{i}{\normalsize %
(x\oplus \alpha }_{j}{\normalsize )\oplus S}_{i}{\normalsize (x)\oplus S}_{k}%
{\normalsize (x+\alpha }_{j}{\normalsize )\oplus S}_{k}{\normalsize (x)}%
\right) \right) .}
\end{equation*}

Of course $n_{ii}=0$. An S-box is said to be good if all off-diagonal
values of its BIC matrix are near to $0.5$. The experimental results of this
test on the newly generated S-boxes ${\small S}_{1667,351}^{N}$,${\small S}%
_{3299,1451}^{D}$ and ${\small S}_{4229,2422}^{M},$ excluding the value $0,$
are shown in a linear order in Figure 3. The minimum, and maximum values of
BIC matrix $N(S)$ of each of the newly generated S-boxes are listed in Table %
\ref{Sbox}. It is evident from Figure 3 and Table \ref{Sbox} that the
S-boxes generated by the proposed methods are strong enough to generate high
diffusion in the data.
\begin{table}[H]
\caption{BIC of the newly generated S-boxes}
\label{Sbox}{\small \centering
\scalebox{0.70} {\begin{tabular}{llllllllllllllll}
\hline \\[-3mm]
S-boxes	&	$S_{1667,351}^{N}$	&	$S_{1949,544}^{N}$	&	$S_{3023,626}^{N}$	&	$S_{3299,1451}^{D}$	&	$S_{3041,1298}^{D}$	&	$S_{3347,2937}^{D}$	&	$S_{4229,2422}^{M}$	&	$S_{4217,1156}^{M}$	&	$S_{3299,1400}^{M}$	\\
\hline																			
BIC(max)	&	0.5273	&	0.5293	&	0.5313	&	0.5371	&	0.5273	&	0.5254	&	0.5254	&	0.5313	&	0.5449	\\
BIC(min)	&	0.4648	&	0.4629	&	0.4707	&	0.4707	&	0.4844	&	0.4746	&	0.4688	&	0.4766	&	0.4727	\\
 \hline
\end{tabular}}  }
\end{table}

\begin{figure}[H]
\label{bic}\centering
\includegraphics [scale=0.80]{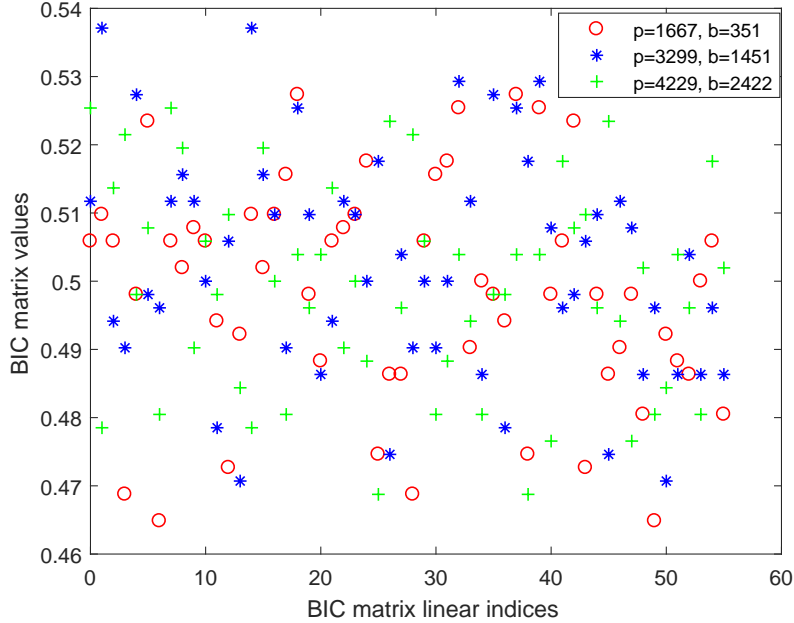}
\caption{BIC matrix plot for ${\protect\small S}_{1667,351}^{N}$,$%
{\protect\small S}_{3299,1451}^{D}$ and ${\protect\small S}_{4229,2422}^{M}$}
\end{figure}

\subsection{Algebraic Complexity(AC)}

The resistance of an S-box against algebraic attacks is measured by
computing its linear polynomial. The AC of an S-box is the number of
non-zero terms in its linear polynomial. The greater is the AC, the greater
is the security of the S-box against algebraic attacks. The AC of the newly generated S-boxes is computed, and is presented in Table %
\ref{ac}. The minimum, and maximum values of AC of the newly generated
S-boxes are $253,$ and $255$, respectively, which are very close to the
optimal value $255$. Thus, the proposed method is able to generate S-boxes
with good AC based on a MEC.

{\small 
\begin{table}[H]
\caption{The AC of the newly generated S-boxes}
\label{ac}{\small \centering
\scalebox{0.70} {\begin{tabular}{llllllllllllllll}
\hline \\[-3mm]
S-boxes	&	$S_{1667,351}^{N}$	&	$S_{1949,544}^{N}$	&	$S_{3023,626}^{N}$	&	$S_{3299,1451}^{D}$	&	$S_{3041,1298}^{D}$	&	$S_{3347,2937}^{D}$	&	$S_{4229,2422}^{M}$	&	$S_{4217,1156}^{M}$	&	$S_{3299,1400}^{M}$	\\
\hline																			
AC	&	254	&	254	&	255	&	255	&	254	&	255	&	253	&	253	&	255	\\
 \hline
\end{tabular}}  }
\end{table}
%
}
\section{Comparison and Discussion}

A detailed comparison of the proposed S-box construction method is performed in this section.

\subsection{Time and Space Complexity}
It is always desirable to have algorithms with low time and space complexity from implementation point of view.
The time and space complexity of the proposed method and other S-box generation methods \cite{Hayat, umar2} based on ECs are compared in Table~\ref{TimeComplexity}.
Note that each method in ~\cite{Hayat, umar2} has quadratic time complexity, while the proposed method takes linear time in the underlying prime $p$ for the generation of an S-box.
However, the space complexity of the methods in ~\cite{Hayat, umar2} is $\mathcal{O}(p)$, where  $p$ is the underlying prime, while it is constant for the proposed method.
Hence, the newly developed method is more suitable for the implementation when compared to all existing S-box generation methods over EC.
\begin{table}[h!]
\caption{Comparison of time and space complexity of the proposed method with other methods over ECs}
\label{TimeComplexity}\centering
{\
\scalebox{0.80} {\begin{tabular}{cccc}
\hline
S-box & Ref. \cite{Hayat} &Ref.  \cite{umar2} & {\small Proposed method} \\
\hline
Time complexity & $\mathcal{O}(p^{2})$ & $\mathcal{O}(p^{2})$ & $\mathcal{O}(p)$  \\
\hline

Space complexity & $\mathcal{O}(p)$ & $\mathcal{O}(p)$ & $\mathcal{O}(1)$  \\ \hline

\end{tabular}%
}}
\end{table}
\subsection{Generation Efficiency}
For a good dynamic S-box construction scheme, it is necessary to ensure the generation of S-box for each valid input parameters, and construct enough number of distinct S-boxes.
It is evident from Theorem~\ref{mrdell} that the proposed method always generate an S-box for each input, while the output of the methods in~\cite{Hayat,umar2} are uncertain i.e., they do not guarantee the construction of S-boxes for each input. This implies that the proposed method is better than the other existing schemes over ECs.

The proposed method can generate at most $p -1$ number of distinct S-boxes for a given prime $p$ and ordering, since for each $b \in [1, p-1] $ it can generate exactly one S-box.
We generated all S-boxes by the proposed method for different primes $p = $ 257, 263, 269, 281, 293, 1013, 1019, 1031, 1049, 1061 and 1997
and each ordering developed in this paper. The number of distinct S-boxes for each ordering is same for all the primes and is listed in listed in Table~\ref{Dsb}.
It is evident from Table~\ref{Dsb} that the number of distinct S-boxes generated by proposed S-box design scheme attains the optimal value and increases with the increase in the size of the prime. Hence, one can generate the desired number of distinct S-boxes by using the proposed method on a appropriate prime.
%

\begin{table}[H]
\caption{The number of distinct  S-boxes constructed by the proposed scheme for some primes}
\label{Dsb}
\scalebox{0.82}{
\begin{tabular}{lcccccccccccc}
\cline{1-12}
$p$           & 257 & 263 & 269          & 281     & 293 &1013&1019&1031&1049, &1061&1997 \\ \hline
 Distinct S-boxes &256    &262     &268     &280      &292& 1012 & 1018 & 1030& 1048 &1060 & 1996 \\ \hline
\end{tabular}
}
\end{table}

\subsection{Cryptographic Properties}
The cryptographic properties of some of the S-boxes constructed by the proposed method are compared with some of the well-known existing S-boxes due to \cite%
{Guoping,Guo,Neural,Shi,Jakimoski,Kim,Hussain1,AES2,Gautam,Chen,Tang} generated by different mathematical structures.
The properties of the S-boxes used in this comparison are listed in Table \ref
{ct}.
Note that the non-linearity (NL) of the S-boxes ${\small S}%
_{1667,351}^{N}$, ${\small S}_{3299,1451}^{D}$ and ${\small S}%
_{4229,2422}^{M}$ is greater than that of the S-boxes in \cite%
{Guoping,Guo,Hayat,Kim,Gautam,Chen,Tang}, and hence the newly generated S-boxes create
better confusion in the data when compared to the later S-boxes.
This implies that the proposed technique is capable of generating S-boxes with good NL when compared to some of the other existing techniques. Moreover, the linear
approximation probability (LAP) of the newly generated S-boxes is better
than the LAP of the S-boxes in \cite{Guoping,Guo,Neural,Gautam,Chen,Tang}, while their differential
approximation probability (DAP) is at most the DAP of the S-boxes in \cite%
{Guoping,Guo,Neural,Hayat,Kim,Gautam,Chen,Tang}.
Thus, the S-boxes generated by the proposed
technique have same or better security against approximation attacks as
compared to the other S-boxes. Similarly, the SAC, BIC and AC test results
of the newly generated S-boxes are comparable with the S-boxes listed in
Table \ref{ct}.
Hence, the proposed S-box generation technique based on a MEC
is capable of generating S-boxes with cryptographic properties comparable with some of the existing S-box construction techniques based on different
mathematical structures.
\begin{table}[H]
\caption{Comparison of the newly generated S-boxes with some of the existing
S-boxes}
\label{ct}{\small \centering
\scalebox{0.87} {\begin{tabular}{lllllllllllllll}
\hline
S-boxes & NL & LAP & DAP & SAC(Max) & SAC(Min) & BIC(Max) & BIC(Min) & AC  \\
\hline

\hline
Ref. \cite{Guoping} & 103 & 0.1328 & 0.0391 & 0.5703 & 0.4414 & 0.5039 & 0.4961 &
255 \\
Ref. \cite{Guo}  & 102 & 0.1484 & 0.0391 & 0.6094 & 0.375 & 0.5215 & 0.4707 & 254
\\
Ref. \cite{Neural} & 106 & 0.1406 & 0.0391 & 0.5938 & 0.4375 & 0.5313 & 0.4648 &
251  \\
Ref. \cite{Hayat} & 104 & 0.0391 & 0.0391 & 0.625 & 0.3906 & 0.53125 & 0.4707 &
255  \\
Ref. \cite{Kim}&  104 & 0.109 & 0.0469 & 0.593 & 0.39 & 0.499 & 0.454 & 255  \\
Ref. \cite{AES2} & 112 & 0.062 & 0.0156 & 0.562 & 0.453 &
0.504 & 0.480 & 9  \\
Ref. \cite{Gautam} & 74 & 0.2109 & 0.0547 & 0.6875 & 0.1094
& 0.5508 & 0.4023 & 253   \\
Ref. \cite{Chen} & 100 & 0.1328 & 0.0547 & 0.6094 & 0.4219
& 0.5313 & 0.4746 & 255   \\
Ref. \cite{Tang}& 103 & 0.1328 & 0.0391 & 0.5703 & 0.3984
& 0.5352 & 0.4727 & 255  \\
$S_{1667,351}^{N}$ &106 & 0.1328 & 0.0391 & 0.5938 & 0.4531 & 0.5273 &
0.4648 & 254   \\
$S_{3299,1451}^{D}$& 106 & 0.1484 & 0.0391 & 0.6406 & 0.4063 & 0.5371 &
0.4707 & 255  \\
$S_{4229,2422}^{M}$& 106 & 0.1328 & 0.0391 & 0.5938 & 0.375 & 0.5254 &
0.4688 & 253   \\
\hline
\end{tabular}}  }
\end{table}

\section{Conclusion}

In this article, we presented an S-box design scheme based on $y$-coordinates of a finite Mordell elliptic curve (MEC), where prime is congruent to $2$ modulo $3$. The technique uses some special type of total
orders on the points of the MEC, and generates an S-box.
The main advantages of the proposed method are that it has linear time complexity, constant space complexity and generate an S-box for each input parameter which are not possible in all existing S-box generation schemes over elliptic curves.
Several standard security tests are
performed on the S-boxes generated by the proposed method to analyze its
cryptographic efficiency.
Experimental results show that the proposed scheme can generate cryptographically strong S-boxes.
Furthermore, it is shown by computational results that the cryptographic properties of the newly generated S-boxes are comparable with some of the well-known existing S-boxes generated by different mathematical structures.

\bibliographystyle{fitee}
\bibliography{OMEC}

\section{Appendix: S-boxes generated by proposed method}
\vspace{-7mm}
\begin{table}[H]
\caption{The S-box ${\protect\small S}_{1667,351}^{N}$ generated by the
proposed method based on the natural ordering}
\label{SN}%
\scalebox{0.5} {\begin{tabular}{llllllllllllllll}
\hline
154 & 217 & 227 & 110 & 85 & 29 & 199 & 37 & 68 & 21 & 91 & 78 & 208 & 3 &
148 & 40 \\
198 & 52 & 54 & 2 & 73 & 7 & 168 & 201 & 229 & 184 & 146 & 6 & 172 & 28 & 44
& 67 \\
195 & 53 & 106 & 10 & 204 & 131 & 157 & 185 & 187 & 156 & 206 & 161 & 81 &
103 & 211 & 33 \\
96 & 159 & 72 & 134 & 164 & 143 & 140 & 193 & 145 & 231 & 237 & 12 & 221 &
188 & 197 & 116 \\
47 & 19 & 129 & 104 & 51 & 236 & 56 & 133 & 55 & 220 & 87 & 1 & 203 & 117 &
210 & 24 \\
4 & 174 & 175 & 113 & 34 & 213 & 171 & 255 & 30 & 43 & 130 & 191 & 57 & 137
& 76 & 234 \\
247 & 244 & 173 & 223 & 63 & 60 & 230 & 166 & 8 & 190 & 139 & 99 & 49 & 200
& 23 & 245 \\
58 & 102 & 226 & 83 & 122 & 70 & 241 & 94 & 127 & 41 & 194 & 233 & 97 & 251
& 107 & 26 \\
109 & 61 & 248 & 90 & 192 & 167 & 147 & 82 & 158 & 225 & 36 & 50 & 84 & 92 &
88 & 38 \\
74 & 136 & 138 & 232 & 62 & 176 & 128 & 189 & 124 & 118 & 169 & 14 & 228 & 0
& 243 & 181 \\
123 & 254 & 20 & 202 & 75 & 149 & 219 & 120 & 160 & 9 & 253 & 39 & 180 & 207
& 114 & 142 \\
183 & 93 & 101 & 15 & 238 & 177 & 132 & 212 & 35 & 250 & 239 & 249 & 179 & 17
& 65 & 186 \\
11 & 125 & 178 & 45 & 170 & 141 & 121 & 126 & 119 & 64 & 144 & 182 & 112 & 22
& 165 & 222 \\
100 & 69 & 252 & 216 & 13 & 27 & 152 & 235 & 80 & 5 & 196 & 59 & 25 & 151 &
79 & 155 \\
240 & 77 & 115 & 71 & 31 & 105 & 95 & 86 & 209 & 150 & 98 & 89 & 163 & 246 &
66 & 18 \\
162 & 214 & 218 & 42 & 242 & 46 & 111 & 48 & 215 & 224 & 135 & 108 & 153 & 32
& 16 & 205 \\ \hline
\end{tabular}}
\end{table}
\vspace{-7mm}
\begin{table}[H]
\caption{The S-box ${\protect\small S}_{3299,1451}^{D}$ generated by the
proposed method based on the diffusion ordering}
\label{SD}%
\scalebox{0.5} {\begin{tabular}{llllllllllllllll}
\hline
33  & 151 & 65  & 207 & 12  & 103 & 96  & 123 & 190 & 126 & 82  & 155 & 21  & 1   & 229 & 186 \\
61  & 224 & 42  & 179 & 63  & 178 & 73  & 153 & 138 & 168 & 146 & 41  & 46  & 9   & 109 & 184 \\
124 & 243 & 236 & 57  & 19  & 6   & 100 & 94  & 69  & 48  & 116 & 216 & 54  & 228 & 90  & 81  \\
47  & 13  & 88  & 197 & 247 & 129 & 206 & 198 & 221 & 5   & 78  & 80  & 150 & 200 & 145 & 55  \\
60  & 105 & 212 & 18  & 210 & 43  & 137 & 250 & 135 & 166 & 52  & 115 & 91  & 208 & 25  & 199 \\
77  & 170 & 121 & 122 & 11  & 254 & 27  & 157 & 175 & 34  & 104 & 201 & 95  & 222 & 133 & 176 \\
36  & 3   & 141 & 218 & 30  & 162 & 220 & 193 & 28  & 110 & 223 & 161 & 74  & 182 & 226 & 113 \\
0   & 112 & 234 & 144 & 241 & 20  & 156 & 62  & 49  & 23  & 26  & 35  & 148 & 101 & 233 & 56  \\
181 & 130 & 118 & 149 & 70  & 173 & 71  & 45  & 50  & 204 & 10  & 87  & 232 & 93  & 177 & 67  \\
4   & 120 & 8   & 40  & 72  & 125 & 92  & 114 & 68  & 83  & 225 & 246 & 158 & 143 & 53  & 196 \\
249 & 242 & 136 & 195 & 160 & 213 & 131 & 107 & 66  & 29  & 230 & 188 & 38  & 111 & 205 & 253 \\
171 & 251 & 102 & 235 & 31  & 127 & 217 & 17  & 183 & 117 & 37  & 211 & 164 & 97  & 119 & 219 \\
167 & 134 & 24  & 16  & 255 & 2   & 32  & 215 & 227 & 154 & 187 & 75  & 231 & 240 & 172 & 142 \\
244 & 89  & 14  & 98  & 76  & 85  & 147 & 79  & 64  & 180 & 214 & 139 & 152 & 238 & 51  & 185 \\
22  & 44  & 194 & 99  & 39  & 169 & 203 & 189 & 108 & 86  & 132 & 237 & 163 & 239 & 209 & 245 \\
59  & 202 & 15  & 58  & 248 & 128 & 174 & 140 & 192 & 191 & 106 & 165 & 159 & 84  & 7   & 252\\ \hline
\end{tabular}}
\end{table}

\begin{table}[H]
\caption{The S-box ${\protect\small S}_{4229,2422}^{M}$ generated by using
the proposed method based on the modulo diffusion ordering}
\label{SM}%
\scalebox{0.5} {\begin{tabular}{llllllllllllllll}
\hline
15  & 13  & 247 & 249 & 167 & 183 & 179 & 173 & 101 & 204 & 105 & 210 & 214 & 205 & 199 & 19  \\
164 & 38  & 85  & 72  & 98  & 90  & 113 & 12  & 239 & 217 & 165 & 228 & 123 & 195 & 26  & 216 \\
207 & 30  & 182 & 219 & 14  & 215 & 232 & 135 & 241 & 145 & 17  & 244 & 223 & 114 & 29  & 70  \\
104 & 81  & 71  & 99  & 191 & 128 & 227 & 86  & 172 & 185 & 5   & 75  & 197 & 184 & 109 & 248 \\
162 & 250 & 25  & 110 & 125 & 230 & 129 & 35  & 102 & 234 & 54  & 171 & 194 & 16  & 33  & 73  \\
155 & 246 & 154 & 84  & 149 & 134 & 238 & 18  & 240 & 67  & 200 & 253 & 61  & 31  & 170 & 180 \\
55  & 20  & 224 & 187 & 10  & 147 & 92  & 133 & 196 & 242 & 146 & 27  & 34  & 140 & 28  & 192 \\
63  & 127 & 143 & 203 & 137 & 2   & 74  & 193 & 65  & 4   & 124 & 51  & 107 & 24  & 42  & 122 \\
103 & 22  & 41  & 226 & 235 & 252 & 116 & 212 & 77  & 49  & 48  & 201 & 148 & 221 & 251 & 80  \\
229 & 115 & 93  & 139 & 181 & 52  & 97  & 119 & 189 & 166 & 21  & 45  & 53  & 100 & 32  & 131 \\
112 & 94  & 59  & 142 & 117 & 36  & 153 & 254 & 66  & 158 & 79  & 121 & 8   & 130 & 132 & 60  \\
245 & 231 & 126 & 152 & 151 & 89  & 0   & 39  & 160 & 136 & 37  & 78  & 236 & 56  & 206 & 157 \\
222 & 174 & 82  & 69  & 6   & 83  & 220 & 3   & 57  & 111 & 208 & 47  & 141 & 87  & 168 & 176 \\
11  & 118 & 169 & 58  & 243 & 120 & 150 & 91  & 190 & 23  & 178 & 44  & 7   & 43  & 177 & 76  \\
161 & 144 & 163 & 68  & 88  & 138 & 218 & 108 & 159 & 186 & 40  & 237 & 175 & 46  & 198 & 96  \\
202 & 9   & 62  & 50  & 64  & 233 & 255 & 209 & 188 & 1   & 106 & 225 & 95  & 213 & 156 & 211\\ \hline
\end{tabular}}
\end{table}

\end{document}